\documentclass[11pt,a4paper]{article}%
\usepackage{amsmath}
\usepackage{amsfonts}
\usepackage{amssymb}
\usepackage{graphicx}
\usepackage{mathrsfs}%
\providecommand{\U}[1]{\protect\rule{.1in}{.1in}}

\providecommand{\U}[1]{\protect\rule{.1in}{.1in}}
\providecommand{\U}[1]{\protect\rule{.1in}{.1in}}
\providecommand{\U}[1]{\protect\rule{.1in}{.1in}}
\newtheorem{theorem}{Theorem}

\newtheorem{lemma}[theorem]{Lemma}

\newtheorem{proposition}[theorem]{Proposition}

\newenvironment{proof}[1][Proof]{\noindent\textbf{#1.} }{\ \rule{0.5em}{0.5em}\vspace{1ex}}
\newcommand{\Lmax}{L_{\max}}
\newcommand{\LmaxOPT}{L_{\max}^{\ast}}

\newcommand{\etal} 	{{\it et al.}}

\begin{document}

\title{Approximation Algorithms for the Open Shop Problem with Delivery Times}
\author{Imed KACEM\thanks{LCOMS, Universit\'{e} de Lorraine, Ile du Saulcy, Metz
57000, France. Contact: \texttt{imed.kacem@univ-lorraine.fr}} and Christophe
RAPINE\thanks{LGIPM, Universit\'{e} de Lorraine, Ile du Saulcy, Metz 57000,
France. Contact : \texttt{christophe.rapine@univ-lorraine.fr}}}
\date{\ \ }
\maketitle

\begin{abstract}
In this paper we consider the open shop scheduling problem where the jobs have
delivery times. The minimization criterion is the maximum lateness of the
jobs. This problem is known to be NP-hard, even restricted to only 2 machines.
We establish that any list scheduling algorithm has a performance ratio of
$2$. For a fixed number of machines, we design a polynomial time approximation
scheme (PTAS) which represents the best possible result due to the strong
NP-hardness of the problem.
\end{abstract}


\section{Introduction}


\textbf{Problem description.} We consider the open shop problem with delivery
times. We have a set $\mathcal{J}=\{1,2,...,n\}$ of $n$ jobs to be performed
on a set of $m$ machines $\mathcal{M}_{1}$, $\mathcal{M}_{2}$, $\mathcal{M}_{3}%
$.... $\mathcal{M}_{m}$.
Each job $j$ consists of exactly $m$ operations $O_{i,j}$ ($i\in
\{1,2,...,m\}$) and has a delivery time $q_{j}$, that we assume non negative.
For every job $j$ and every index $i$, operation $O_{i,j}$ should be performed
on machine $\mathcal{M}_{i}$. The processing time of each operation $O_{i,j}$
is denoted by $p_{i,j}$.
At any time, a job can be processed by at most one machine. Moreover, any
machine can process only one job at a time. Preemption of operations is not
allowed. We denote by $C_{i,j}$ the completion time of operation $O_{i,j}$.
For every job $j $, its completion time $C_{j}$ is defined as the completion
time of its last operation. The lateness $L_{j}$ of job $j$ is equal to
$C_{j}+q_{j}$. The objective is to find a feasible schedule that minimizes the
maximum lateness $L_{\max}$, where%

\begin{equation}
L_{\max}=\max_{1\leq j\leq n}\left\{  L_{j}\right\}
\end{equation}

For any feasible schedule $\pi$, we denote the resulting maximum lateness by
$L_{\max}\left(  \pi\right)  $. Moreover, $L_{\max}^{\ast}$ denotes the
maximum lateness of an optimal solution $\pi^{\ast}$, that is, $L_{\max}%
^{\ast}=$ $L_{\max}\left(  \pi^{\ast}\right)  $. According to the tertiary
notation, the problem is denoted as $O||L_{\max}$.


Recall that a constant approximation algorithm of performance ratio $\gamma
\ge1$ (or a $\gamma$-approximation) is a polynomial time algorithm that
provides a schedule with maximum lateness no greater than $\gamma L_{\max
}^{\ast}$ for every instance. A polynomial time approximation scheme (PTAS) is
a family of ($1+\varepsilon$)-approximation algorithms of a polynomial time
complexity for any fixed $\varepsilon>0$. If this time complexity is
polynomial in $1/\varepsilon$ and in the input size then we have a fully
polynomial time approximation scheme (FPTAS).

\vspace{1ex}

\textbf{Related approximation results.}
According to the best of our knowledge, the design of approximation algorithms
has not yet been addressed for problem $O||L_{\max}$. However, some
inapproximability results have been established in the literature. For a fixed
number of machines, unless \texttt{P=NP}, problem $Om||L_{\max}$ cannot admit
an FPTAS since it is \texttt{NP}-hard in the strong sense on two
machines~\cite{Lawler},~\cite{Lawler2}. The existence of a PTAS for a fixed
$m$ is an open question, that we answer positively in this paper. If the
number $m$ of machines is part of the inputs, Williamson \textit{et
al}~\cite{Williamson} proved that no polynomial time approximation algorithm
with a performance guarantee lower than $5/4$ can exist, , unless \texttt{P=NP}, which precludes the existence of a PTAS.
Several interesting results exist for some related problems, mainly to
minimize the makespan :

\begin{itemize}
\item[$\bullet$] Lawler \textit{et al}~\cite{Lawler}-\cite{Lawler2} presented
a polynomial algorithm for problem $O2|\mbox{\it pmtn}|L_{\max}$. In contrast,
when preemption is not allowed, they proved that problem $O2||L_{\max}$ is
strongly \texttt{NP}-hard, as mentioned above.

\item[$\bullet$] Gonzales and Sahni \cite{Gonzales} proved that problem
$Om||C_{\max}$ is polynomial for $m=2$ and becomes \texttt{NP}-hard when
$m\geq3$ .

\item[$\bullet$] Sevastianov and Woeginger \cite{Sevastainov} established the
existence of a PTAS for problem $Om||C_{\max}$ when $m$ is fixed.

\item[$\bullet$] Kononov and Sviridenko~\cite{Kononov} proposed a PTAS for
problem $Oq(Pm)|r_{ij}|C_{\max}$ when $q$ and $m$ are fixed.

\item[$\bullet$] Approximation algorithms have been recently proposed for
other variants such as the two-machine routing open shop problem. A sample of
them includes Chernykh \textit{el al}~\cite{Chernykh} and Averbakh \textit{et
al}~\cite{Averbakh}.
\end{itemize}
\
Finally, we refer to the state-of-the-art paper on scheduling problems under the maximum lateness minimization by Kellerer \cite{Kelate}.
\vspace*{1ex}

\textbf{Contribution.} Unless \texttt{P=NP}, problem $Om||L_{\max}$ cannot
admit an FPTAS since it is \texttt{NP}-hard in the strong sense on two
machines. Hence, the best possible approximation algorithm is a PTAS. In this
paper, we prove the existence of such an algorithm for a fixed number of
machines, and thus gives a positive answer to this open problem. Moreover, we
provide the analysis of some simple constant approximation algorithms when the
number of machines is a part of the inputs.

\vspace*{1ex}

\textbf{Organization of the paper.} Section 2 present some simple preliminary
approximation results on list scheduling algorithms. In Section 3, we describe
our PTAS and we provide the analysis of such a scheme. Finally, we give some
concluding remarks in Section 4.

\section{Approximation Ratio of List Scheduling Algorithms}


List scheduling algorithms are popular methods in scheduling theory. Recall
that a list scheduling algorithm relies on a greedy allocation of the
operations to the resources that prevents any machine to be inactive while an
operation is available to be performed. If several operations are concurrently
available, ties are broken using a priority list. We call a \textit{list
schedule} the solution produced by a list scheduling algorithm. We establish
that any list scheduling algorithm has a performance guarantee of $2$,
whatever its priority rule. Our analysis relies on $2$ immediate lower bounds,
namely the conservation of the work and the critical path. Let us denote
\[
P=\max_{i=1,...,m}\{\sum_{j=1}^{n}p_{ij}\}\ \mbox{ and }\ Q=\max
_{j=1,...,n}\{\sum_{i=1}^{m}p_{ij}+q_{j}\}
\]
Clearly $L_{\max}^{\ast}\geq P$ and $L_{\max}^{\ast}\geq Q$. We have the
following result :

\begin{proposition}
\label{Proposition:list} Any list scheduling algorithm is a $2$-approximation
algorithm for problem $O||L_{\max}$. More precisely, for any list schedule
$\pi$, $L_{\max}(\pi) \le P + Q$
\end{proposition}

\begin{proof}
Consider a list schedule $\pi$, and let $u$ be a job such that $L_{u}=\Lmax(\pi)$. Without loss of generality, we can assume that the last operation of $u$ is scheduled on the first machine. We consider $2$ cases : either an
idle-time occurs on $\mathcal{M}_{1}$ before the completion of job $u$, or
not. If there is no idle time on $\mathcal{M}_{1}$, then $L_{u}\leq
P+q_{u}\leq P+Q$. Otherwise, let us denote by $I$ the total idle time occuring
on $\mathcal{M}_{1}$ before the completion time of job $u$. We have $L_{u}\leq
P+I+q_{u}$. Notice that job $u$ could not have been available on machine
$\mathcal{M}_{1}$ at any idle instant, otherwise, due to the principle of list
scheduling algorithms, it would have been scheduled. As a consequence, an
operation of job $u$ is performed on another machine at every idle instant of
$\mathcal{M}_{1}$ before $C_{u}$. Hence, we can bound the idle time $I$ by the
total processing time of job $u$. We have :
\[
L_{u}\leq P+I+q_{u}\leq P+\sum_{i=1}^{m}p_{iu}+q_{u}\leq P+Q
\]
We can conclude that in any case $L_{\max}(\pi)\leq P+Q \le 2\LmaxOPT$
\end{proof}

Notice that good \textit{a posteriori} performances can be achieved by a list
scheduling algorithm, for instance if the workload $P$ is large compared with
the critical path $Q$. One natural question is whether some better
approximation ratios can be obtained with particular lists. It is a folklore
that minimizing the maximum lateness on one ressource can be achieved by
sequencing the tasks in non-increasing order of their delivery times. This
sequence is known as Jackson's order. One can wonder if a list scheduling
algorithm using Jackson's order as its list performed better in the worst
case. The answer is negative. The following proposition states that the
analysis of Proposition~\ref{Proposition:list} is tight whatever the list.

\begin{proposition}
No list scheduling algorithm can have a performance ratio less than $2$ for
problem $O2||L_{\max}$.
\end{proposition}

\begin{proof}
Consider the following instance: we have $3$ jobs to schedule on $2$ machines.
Jobs $1$ and $2$ have only one (non null) operation to perform, respectively
on machine $\mathcal{M}_{1}$ and $\mathcal{M}_{2}$. The duration of the
operation is equal to $a$ time units, where $a\geq1$ is a parameter of the
instance. Both delivery times are null. Job $3$ has one unit operation to
perform on each machine, and its delivery time is $q_{3}=a$.

An optimal schedule sequences first Job $3$ on both machines, creating an idle
time at the first time slot, and then performs Jobs $1$ and $2$. That is, the
optimal sequence is $(3,1)$ on $\mathcal{M}_{1}$ and $(3,2)$ on $\mathcal{M}%
_{2}$. The maximum lateness is equal to $L_{\max}^{\ast}=a+2$. Notice that
this schedule cannot be obtained by a list scheduling algorithm, since an idle
time occurs at the first instant while a job (either $1$ or $2$) is available.
Indeed, it is easy to see that, whatever the list, either Job $1$ or Job $2$
is scheduled at time $0$ by a list scheduling algorithm. As a consequence, Job
$3$ cannot complete before time $a+1$ in a list schedule $\pi$. Hence,
$L_{\max}(\pi)\geq2a+1$. The ratio for this instance is $\frac{2a+1}{a+2}$,
which tends to $2$ when $a$ tends to $+\infty$.
\end{proof}

\begin{figure}
\begin{center}
    \includegraphics[width=1.3\textwidth]{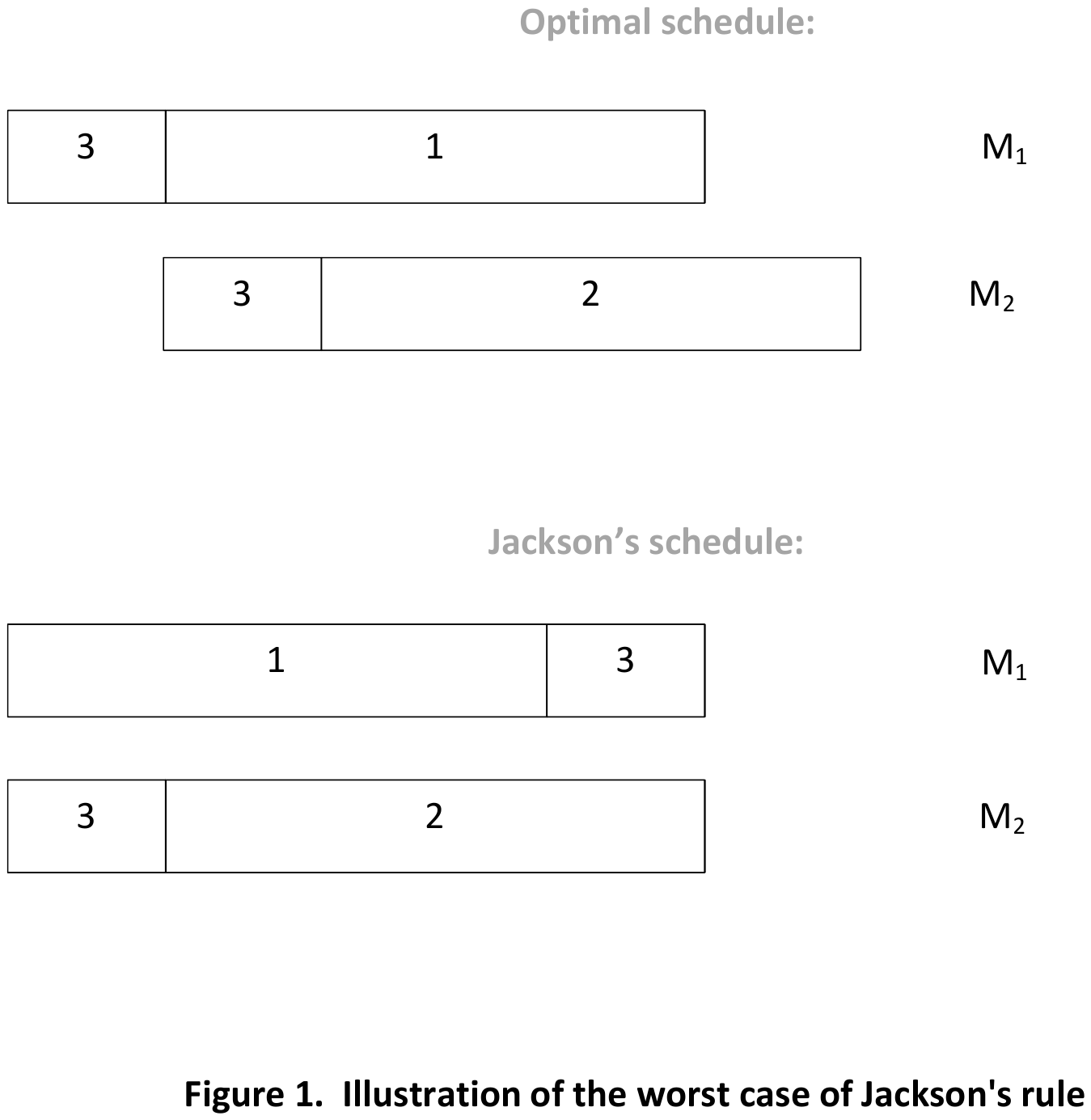}
\end{center}
\end{figure}


As a conclusion, Jackson's list does not perform better that any other list in
the worst case. Nevertheless, we use it extensively in the PTAS that we
present in the next section.

\section{PTAS}



In this section, we present the first PTAS for problem $Om||L_{\max}$, that
is, when the number of machines is fixed. Our algorithm considers three
classes of jobs as introduced by Sevastianov and Woeginger~\cite{Sevastainov}
and used by several authors for a variety of makespan minimization in shops
(see for instance the extension by Jansen \etal for the job shop~\cite{Jansen}%
). Notice that our approximation algorithm does not require to solve any linear program.

\subsection{Description of the Algorithm}

Let $\varepsilon$ be a fixed positive number. We describe how to design an
algorithm, polynomial in the size of the inputs, with a performance ratio of
$(1+\varepsilon)$ for problem $Om||L_{\max}$. As a shorthand, let
$\overline{\varepsilon}=\frac{\varepsilon}{2m(m+1)}$. Recall that
$P=\max_{i=1}^{m}\{\sum_{j=1}^{n}p_{ij}\}$ is the maximal workload of a
machine. For a given integer $k$, we introduce the following subsets of jobs
$\mathcal{B}$, $\mathcal{S}$ and $\mathcal{T}$:%

\begin{align}
\mathcal{B}  &  =  \left\{  j\in\mathcal{J}\ |\ \max_{i=1}^{m} p_{i,j}%
\geq\overline{\varepsilon}^{k}P\right\} \\
\mathcal{S}  &  =  \left\{  j\in\mathcal{J} \ |\ \overline{\varepsilon}^{k}P >
\max_{i=1}^{m} p_{i,j}\geq\overline{\varepsilon}^{k+1}P\right\} \\
\mathcal{T}  &  =  \left\{  j\in\mathcal{J} \ |\ \overline{\varepsilon}^{k+1}P
> \max_{i=1}^{m} p_{i,j}\right\}
\end{align}
%

By construction, for any integer $k$, sets $\mathcal{B}$, $\mathcal{S}$ and $\mathcal{T}$ define a partition of the jobs.
For the ease of understanding, the jobs of $\mathcal{B}$ will be often called
the \textit{big} jobs, the jobs of $\mathcal{S}$ the \textit{small} jobs, and
the jobs of $\mathcal{T}$ the \textit{tiny} jobs. Notice that the duration of
any operation of a small jobs is less than $\overline{\varepsilon}^{k}P$, and
less than $\overline{\varepsilon}^{k+1}P$ for a tiny job. The choice of $k$
relies on the following proposition, which comes from Sevastianov and Woeginger~\cite{Sevastainov}:

\begin{proposition} ~\cite{Sevastainov}
\label{Proposition:k} There exists an integer $k\leq\lceil\frac{m}%
{\varepsilon}\rceil$ such that
\begin{equation}
p(\mathcal{S})\leq\varepsilon P\label{cond}%
\end{equation}
where $p(\mathcal{S})=\sum_{j\in\mathcal{S}}\sum_{i=1}^{m}p_{ij}$ is the total
amount of work to perform for the jobs of $\mathcal{S}$. Moreover, for the big jobs, we have:
\begin{equation}
\left\vert \mathcal{B}\right\vert \leq\frac{m}{\overline{\varepsilon}^{k}%
}\label{cond2}%
\end{equation}

\end{proposition}

\begin{proof}
Let us denote $z=\lceil m/{\varepsilon}\rceil$. Observe that for a given value
$k$, the duration of the largest operation of any small job belongs to the
interval $I_{k}=[\overline{\varepsilon}^{k+1}P,\overline{\varepsilon}^{k}P[$.
Assume for the sake of contradiction that, for all values $k=1,\dots,z$, the
corresponding set $\mathcal{S}_{k}$ does not verify Condition (\ref{cond}). As
a consequence, $p(\mathcal{S}_{k})>\varepsilon P$ for each $k=1,\dots,z$.
Since these sets are disjoint, it results that the total processing time of
the operations of the jobs whose the duration of its largest operation belongs
to $\in\left[  \overline{\varepsilon}^{z+1}P,P\right[  $ is strictly greater
than $z\varepsilon P$. However, this amount of work is bounded by the total
work of the instance. We have :
\[
z\varepsilon P<\sum_{i=1}^{m}\sum_{j=1}^{n}p_{ij}\leq mP
\]
Thus $z<m/\varepsilon$, which contradicts our definition of $z$. It follows
that at least one interval $\mathcal{I}_{k}\mathcal{=}[\overline{\varepsilon
}^{k+1}P,\overline{\varepsilon}^{k}P[$ with $1\leq k\leq z$ is suitable to
contain the values of the large operations of subset $\mathcal{S}$ such that
$p(\mathcal{S})\leq\varepsilon P$.

To prove Inequality (\ref{cond2}), we can observe that the total processing
time of the operations of $\mathcal{B}$ is bounded by $mP$. Thus, $\left\vert
\mathcal{B}\right\vert \overline{\varepsilon}^{k}P\leq mP$ must hold and
Inequality (\ref{cond2}) follows.
\end{proof}

Notice that, for a fixed value $m$ of machines, only a constant number $\lceil
m/{\varepsilon}\rceil$ of values must be considered for $k$. Hence, an integer
$k$ verifying the conditions of Proposition~\ref{Proposition:k} can be found
in linear time. Assume from now that $k$ has been chosen according to
Proposition~\ref{Proposition:k}. In order to present our approach, let us
explain how the different sets $\mathcal{S}$, $\mathcal{B}$ and $\mathcal{T}$
of jobs are scheduled in our PTAS. Since set $\mathcal{S}$ represents a very
small work, we can schedule it first. Clearly, its last operation cannot
complete after time $t(\mathcal{S})\leq\varepsilon P$ in a list schedule.
Since set $\mathcal{B}$ has a fixed number of jobs, we can afford to consider
all the ways to sequence them. For that, we discretize the time, considering a
time step $\delta=\overline{\varepsilon}^{k+1}P$. Finally, for each assignment
of the big jobs, we schedule the tiny jobs using simply Jackson's list
scheduling algorithm. One originality of our approach is the possibility for a
tiny job to push a big job in order to fit before it. More precisely, if the
tiny job the list scheduling algorithm is considering cannot complete before
the start of the next big job on its machine, say $b$, then we force its
schedule by shifting right the operation of job $b$ as much as necessary. This
shifting is special in twofolds : first, we also shift right of the same
amount of time \textit{all} the operations of the big jobs starting after job
$b$. Second, the operation of job $b$ is then \textit{frozen}, that is, it
cannot be pushed again by a tiny job. Hence, an operation of a big job can be
pushed at most once by a tiny job, but can be shifted right a lot of times,
due to the push of other operations of some big jobs. A more formal
description of our algorithm can be given as follows:

\vspace*{1ex}

\texttt{ALGORITHM PTAS}

\begin{enumerate}
\item Schedule first jobs of $\mathcal{S}$ using any list scheduling algorithm
between time $0$ to time $p(\mathcal{S})$ (the cost factor of this
simplification will not be more than $1+\varepsilon$).

\item Let $\delta=\overline{\varepsilon}^{k+1}P$. Consider all the time
intervals between $p(\mathcal{S})$ and $mP$ of length $\delta$ (the number of
these intervals is a constant for a fixed $\varepsilon$).

\item Enumerate all the schedules of jobs in $\mathcal{B}$ between
$p(\mathcal{S})$ and $mP$. Here, a schedule is reduced to an assignment of the
operations to starting times of the time intervals defined in the previous
step (the cost factor of this simplification will not be more than
$1+\varepsilon$).

\item Complete every partial schedule generated in the last step by adding the
jobs of $\mathcal{T}$. The operations of $\mathcal{T}$ are added by applying a
list scheduling algorithm using Jackson's order (i.e., when several operations
are available to be performed we start by the one of the largest delivery
time). Note that if an operation cannot fit in front of a big job $b$, then we
translate $b$ and all the next big jobs by the same necessary duration to make
the schedule feasible. The operation of job $b$ is then frozen, and cannot be shifted any more.

\item Return the best feasible schedule found by the algorithm.
\end{enumerate}

\subsection{Analysis of the Algorithm}


We start by introducing some useful notations. Consider a schedule $\pi$. For
each machine $i$, we denote respectively by $s_{ir}$ and $e_{ir}$ the start
time and completion time of the $r$th operation of a big job on machine $i$,
for $r=1,\dots,|\mathcal{B}|$. By convenience we introduce $e_{i0}=0$. For
short we call the \textit{grid} the set of all the couples $($resource
$\times$ starting time$)$ defined in Phase $(2)$ of the algorithm. Recall that in
the grid the starting times are discretized to the multiples of $\delta$.
Notice that our algorithm enumerates in Phase $(3)$ all the assignments of
big job operations to the grid. Phase $(4)$ consists in scheduling all the
tiny jobs in-between the big jobs. In the following, we call a time-interval
on a machine corresponding to the processing of a big job a \textit{hole}, for
the machine is not available to perform the tiny jobs. The duration of
the $r$th hole on machine $i$, that is $e_{ir}-s_{ir}$, is denoted by $h_{ir}%
$. By analogy to packing, we call a \textit{bin} the time-interval between two
holes. The duration of the $r$th bin on machine $i$, that is $s_{ir}%
-e_{i,r-1}$, is denoted by $a_{ir}$. We also introduce $H_{ir}=h_{i1}%
+\dots+h_{ir}$ and $A_{ir}=a_{i1}+\dots+a_{ir}$, that is the overall duration
of the $r$ first holes and bins, respectively, on machine $i$.\newline

Now consider an optimal schedule $\pi^{\ast}$. With immediate notations, let
$s_{ir}^{\ast}$ be the start time of the $r$th operations of a big job on the
machine $i$, and let $A_{ir}^{\ast}$ be the overall duration of the $r$ first
bins. For the ease of the presentation, we assume in the reminder, without
loss of generality, that we have no small jobs to schedule : Indeed,  Phase (1) does not increase the length of
the schedule by more than $\varepsilon P\leq\varepsilon \LmaxOPT$. We
say that an assignment to the grid is \textit{feasible} if it defines a
feasible schedule for the big jobs. The next lemma shows that there exists a
feasible assignment such that each operation of the big jobs is delayed,
compared to an optimal schedule, by at least $2m\delta$ time units and by at
most $(2+|\mathcal{B}|)m\delta$ time units.

\ \ 

\begin{lemma}
\label{lemma:affectation} There exists a feasible assignment $\bar{s}$ to the
grid such the operations of the big jobs are sequenced in the same order, and
for every machine $i$ and index $r$ we have:
\[
s_{ir}^{\ast}+2m\delta\leq\bar{s}_{ir}\leq s_{ir}^{\ast}+(2+|\mathcal{B}%
|)m\delta
\]

\end{lemma}


\begin{proof}
Among all the possible assignments enumerated in Phase $(3)$ for the big
jobs, certainly we consider the following one, which corresponds to a shift of
the optimal schedule $\pi^{\ast}$ restricted to the big jobs :

\begin{itemize}
\item Insert $2m\delta$ extra time units at the beginning of $\pi^{*}$, that
is delay all the operations by $2m\delta$.

\item Align the big jobs to the grid (shifting them to the right)

\item Define the assignment $\bar{s}$ as the current starting times of the
operations of the big jobs.

\end{itemize}

More precisely, to align the big jobs to the grid, we consider sequentially
the operations by non-decreasing order of their starting time. We then shift
right the current operation to the next point of the grid and translate the
subsequent operations of the same amount of time. This translation ensures
that the schedule remains feasible for the big jobs.

By construction each operation is shifted right by at least $2m\delta$ time
units, which implies that $\bar{s}_{ir}\geq s_{ir}^{\ast}+2m\delta$. The
alignment of an operation to the grid again shifts it right, together with all
the subsequent operations, by at most $\delta$ time units. Thus, the last
operation is not shifted more than $m|\mathcal{B}|\delta$ time units by the
alignment. The result follows.
\end{proof}

Now consider the schedule $\pi$ obtained by applying the Jackson's list
scheduling algorithm to pack the tiny jobs between the holes, starting from
the feasible assignment of Lemma~\ref{lemma:affectation}. Notice that, due to
the shift procedure in Phase (4), the starting time of the big jobs (the
holes) can change between the assignment $\bar{s}$ and the schedule $\pi$.
However a hole can be shifted at most $m|\mathcal{B}|$ times since each
operation of a big job is shifted at most once by a tiny job. Moreover the length of a shift
is bounded by the duration of an operation of a tiny job, that is by $\delta$.
In addition, as we shift all the operations belonging to the big jobs, the
length of the bins cannot decrease in the schedule $\pi$. Hence, we have the
two following properties for the schedule $\pi$, which are direct consequences
of Lemma~\ref{lemma:affectation} and of the previous discussion :

\begin{enumerate}
\item Any operation of a big job is only slightly delayed compared to the
optimal schedule $\pi^{\ast}$ : $s_{ir}\leq s_{ir}^{\ast}+2(|\mathcal{B}%
|+1)m\delta$

\item Each bin is larger in $\pi$ than in the optimal schedule. More precisely
we have $A_{ir} \ge A^{*}_{ir}+2m\delta$ for all machine $i$ and all index $r$.
\end{enumerate}

In other words in the schedule $\pi$ we have slightly delayed the big jobs to
give more room in each bin for the tiny jobs. We say that a job $y$ is more
critical than a job $x$ if $y$ has a higher priority in the Jackson's order.
By convention a job is as critical at itself. We have the following lemma:

\begin{lemma}
\label{Lemma:critical} \label{lemma:C} In schedule $\pi$, for every job $x$,
there exists a job $y$ such that :
\[
q_{y}\geq q_{x} \ \mbox{ and } C_{x}\leq C_{y}^{\ast}+2(|\mathcal{B}%
|+1)m\delta
\]

\end{lemma}

\begin{proof}
Let $x$ be a job and $C_{x}$ its completion time in schedule $\pi$. Without
loss of generality we can assume that the last operation of the job $x$ is
processed on the first machine. If $x$ is a big job, that is $x\in\mathcal{B}%
$, we have already noticed that we have $C_{x}\leq C_{x}^{\ast}+2(|\mathcal{B}%
|+1)m\delta$, due to our choice of the big jobs assignment on the grid. Hence,
the inequality of Lemma~\ref{lemma:C} holds for $x$. Thus consider in the
remaining of the proof the case of a tiny job $x$. We denote by $\mathcal{T}%
_{x}$ the subset of tiny jobs that are more critical than $x$ and such that
their operation on the first machine is completed by time $C_{x}$, that is :
\[
\mathcal{T}_{x}=\{y\in\mathcal{T}\ |\ C_{y}^{1}\leq C_{x}^{1}\mbox{ and }q_{y}%
\geq q_{x}\ \}
\]
Observe that our definition implies in particular that $x\in\mathcal{T}_{x}$.
We first establish that in schedule $\pi$, almost all the tiny jobs
processed before $x$ on the first machine are more critical than $x$. That is,
the schedule $\pi$ essentially follows the Jackson's sequence for the tiny
jobs. Let $r$ be the index of the bin where $x$ completes in schedule $\pi$. For short we denote by $A_{1}(x)$ the overall time available for
processing tiny jobs on the first machine over the time-interval $[0,C(x)]$, that is $A_{1}(x)=C_{x}-H_{1,r}$. 
We also denote by $p_{1}(\mathcal{T}_{x})$ 
the total processing time of the operations 
of $\mathcal{T}_{x}$ on the first machine. We claim that :

\begin{equation}
p_{1}(\mathcal{T}_{x})\geq A_{1}(x)-2(m-1)\delta\label{eq:p(Sx)}%
\end{equation}
 
If at every available instant on the first machine till the completion of $x$
an operation of $\mathcal{T}_{x}$ is processed in $\pi$, then clearly we have
$A_{1}(x)=p_{1}(\mathcal{T}_{x})$ and Inequality~\ref{eq:p(Sx)} holds.
Otherwise, consider a time interval $I=[t,t^{\prime}]$, included in a bin, such
that no task of $\mathcal{T}_{x}$ is processed. We call such an interval
\textit{non-critical} for $x$. It means that during $I$, either some idle
times appear on the first machine, and/or some jobs less critical than $x$
have been processed. However, due to the shift procedure and the Jackson's
list used by the algorithm, the only reason for not scheduling $x$ during $I$
is that this job is not available by the time another less critical job $z$
is started. 
Notice that in an open-shop environment, a job $x$ is not available on the first machine only if one of
its operations is being processed on another machine. As a consequence, the
interval $I$ necessarily starts during the processing of $x$ on another
machine, that is $t\in\lbrack t_{i,x},C_{i,x}]$ for some machine $i$. This
holds for any idle instant and any time an operation is started in interval
$I$. As a consequence, the interval $I$ cannot finish later than the
completion of $x$ on another machine $i^{\prime}$, plus the duration of a less
critical (tiny) job $z$ eventually started on the first machine during time
interval $[t_{i^{\prime},x},C_{i^{\prime},x}]$. Since all the jobs are tiny,
the overall duration of the non-critical intervals for $x$ is thus bounded by
$\sum_{i=2}^{m}(p_{i,x}+\delta)$, which is at most equal to $2(m-1)\delta$.
Inequality~\ref{eq:p(Sx)} follows.\newline

Now let $y$ be the job of $\mathcal{T}_{x}$ that completes last on the first
machine in the optimal schedule $\pi^{\ast}$. Let $r^{\ast}$ be the index of
the bin where $y$ is processed on the first machine in $\pi^{\ast}$, and let
$A_{1}^{\ast}(y)$ be the total available time for tiny jobs in $\pi^{\ast}$
before time $C_{1,y}^{\ast}$, that is $A_{1}^{\ast}(y)=C_{1,y}^{\ast
}-H_{1,r^{\ast}}^{\ast}$. Recall that $r$ is the number of bins used in 
schedule $\pi$ to process all the operations of $\mathcal{T}_{x}$ on the
first machine. We prove that the optimal schedule also uses (at least) this
number of bins, that is $r^{\ast}\geq r$. Indeed, by the conservation of work,
we have that $A_{1}^{\ast}(y)\geq p_{1}(\mathcal{T}_{x})$. Using
Inequality~\ref{eq:p(Sx)} we obtain that $A_{1}^{\ast}(y)\geq A_{1}%
(x)-2(m-1)\delta$. By definition of $r$ and $r^{\ast}$, we also have
$A_{1,r-1}\leq A_{1}(x)$ and $A_{1}^{\ast}(y)\leq A_{1,r^{\ast}}^{\ast}$.
Hence, the following inequality must hold:
\[
A_{1,r-1}\leq A_{1,r^{\ast}}^{\ast}+2(m-1)\delta
\]
However, we have observed that our choice of the assignment of the big jobs
to the grid ensures that for any index $l$, $A_{1,l}^{\ast}+2m\delta\leq
A_{1,l}$, which implies that we have $A_{1,r-1}+2\delta\leq
A_{1,r^{\ast}}$. As a consequence, inequality $A_{1,r-1}<A_{1,r^{\ast}}$ must
hold. Since $A_{1,l}$ represents the total length of the $l$ first bins in
$\pi$, which is obviously non-decreasing with $l$, it implies that $r\leq
r^{\ast}$.

It means that in $\pi^{\ast}$, task $y$ cannot complete its operation on
the first machine before the first $r$ big tasks. We can conclude the proof of
Lemma~\ref{lemma:C} by writing that, on one hand, $C_{y}^{\ast}\geq
p_{1}(\mathcal{T}_{x})+H_{1,r}$, and, on the other hand, $C_{x}\leq
p_{1}(\mathcal{T}_{x})+2(m-1)\delta+H_{1,r}$. As a consequence, $x$ does not
complete in $\pi$ latter than $2(m-1)\delta$ times units after the completion
time of $y$ in $\pi^{\ast}$. Since by definition $y$ is more critical than
$x$, Lemma~\ref{lemma:C} follows.
\end{proof}

Finally, we can conclude that the following theorem holds:

\begin{theorem}
Problem $Om||L_{\max}$ admits a PTAS.
\end{theorem}

\begin{proof}
We first establish that the maximum lateness of the schedule returned by our
algorithm is bounded by $(1+\varepsilon)\LmaxOPT$. In schedule $\pi$
defined in Lemma \ref{lemma:C}, let $u$ be a job such that $\Lmax(\pi)=C_{u}+q_{u}$. 
If job $u$ is a small job, then it completes before time $p(\mathcal{S})$. Due to our choice of the partition, see Proposition~\ref{Proposition:k},
we have $L_u \le \varepsilon P + q_u \le (1+\varepsilon)\LmaxOPT$. Hence, in the following, we restrict to the case 
where $u \notin \mathcal{S}$, that is, job $u$ is either a big or a tiny job. According to Lemma~\ref{lemma:C}, there exists a job $y$ such that :
\[
q_{u}\leq q_{y}\ \mbox{ and }\ C_{u}\leq C_{y}^{\ast}+2(|\mathcal{B}%
|+1)m\delta
\]
We have :
\begin{align*}
L_{\max}\left(  \pi\right)   &  =C_{u}+q_{u}\\
&  \leq C_{y}^{\ast}+2(|\mathcal{B}|+1)m\delta+q_{y}\\
&  \leq L_{\max}^{\ast}+2(|\mathcal{B}|+1)m\delta
\end{align*}

As a consequence, using Proposition~\ref{Proposition:k}, we can write that for
any fixed $\varepsilon\le1$ :
\begin{align*}
L_{\max}\left(  \pi\right)  -L_{\max}^{\ast}  &  \leq2(\frac{m}{\overline
{\varepsilon}^{k}}+1)m\overline{\varepsilon}^{k+1}P\\
&  \leq2(\frac{m+1}{\overline{\varepsilon}^{k}})m\overline{\varepsilon}%
^{k+1}P\\
&  =2(m+1)m\overline{\varepsilon}P\\
&  \leq\varepsilon L_{\max}^{\ast}%
\end{align*}

Hence, our algorithm has a performance guarantee of $(1+\varepsilon)$. Let us
now check its time complexity. First, the identification of $k$ and the three
subsets $\mathcal{B}$, $\mathcal{S}$ and $\mathcal{T}$ can be done in
$O(\frac{m^{2}.n}{\varepsilon})$. Second, the scheduling of the jobs of $\mathcal{S}$ can clearly be performed in polynomial time 
(in fact, in linear time in $n$ for $m$ fixed). Now, let us consider the scheduling of the big jobs.
The number $\Delta$ of points in the grid is bounded by:
\begin{align*}
\Delta & \le m \times \frac{mP}{\delta} \ =\  \frac{m^{2}}{\overline{\varepsilon}^{k+1}} \\
       & \le  \frac{m^{2}}{\overline{\varepsilon}^{2+\frac{m}{\varepsilon}}} \\
       & \le m^{2}\left( \frac{2m(m+1)}{\varepsilon} \right)^{2+\frac{m}{\varepsilon}}
\end{align*}
The second inequality comes from the fact that $k \le \lceil m/\varepsilon \rceil$, due to Proposition~\ref{Proposition:k}.
The last bound is clearly a constant for $m$ and $\varepsilon$ fixed. 
The number of possible assignments of jobs of $\mathcal{B}$ in Phase
$(3)$ is bounded by

\begin{align*}
(m| \mathcal{B} |) ^ \Delta & \le \left( \frac{m^2}{\overline{\varepsilon}^{k}} \right)^\Delta   \\
\end{align*}
Hence, only a constant number of assignments to the grid are to be considered.
%
%
Phase $(4)$ completes every feasible assignment in a polynomial time. Phase
$(5)$ outputs the best solution in a linear time of the number of feasible
assignments. In overall, the algorithm is polynomial in the size of the
instance for a fixed $\varepsilon$ and a fixed $m$.
\end{proof}

\section{Conclusion}

In this paper we considered an open question related to the existence of PTAS
to the $m$-machine open shop problem where $m$ is fixed and the jobs have
different delivery times. We answered successfully to this important question.
This represents the best possible result we can expect due to the strong
NP-hardness of the studied problem.

Our perspectives will be focused on the study of other extensions. Especially,
the problem with release dates seems to be very challenging.

\end{document}